\documentclass[sn-mathphys,Numbered]{sn-jnl}

\usepackage{graphicx}
\usepackage[latin1]{inputenc}
\usepackage{amssymb,amsmath,array}
\usepackage{adjustbox}
\usepackage{amsthm}
\usepackage{url}
\usepackage{orcidlink}
\usepackage{mathtools}
\usepackage[algo2e, linesnumbered, ruled, vlined]{algorithm2e}
\usepackage{algorithm}
\usepackage{algorithmic}
\DeclarePairedDelimiter\ceil{\lceil}{\rceil}
\usepackage{multirow}%
\usepackage{amsthm}%
\usepackage{mathrsfs}%
\usepackage[title]{appendix}%
\usepackage{xcolor}%
\usepackage{textcomp}%
\usepackage{amssymb}
\usepackage{manyfoot}%
\usepackage{booktabs}%
\usepackage[algo2e, ruled, vlined]{algorithm2e}
\usepackage{algorithm}
\usepackage{algorithmic}
\usepackage{listings}
\usepackage{tikz}%
\usetikzlibrary{quantikz}

\newtheorem{theorem}{Theorem}
\newtheorem{proposition}[theorem]{Proposition}

\begin{document}

\title[Quantum Subroutine for Variance Estimation: Algorithmic Design and  Applications]{Quantum Subroutine for Variance Estimation: Algorithmic Design and Applications}

\author[1]{\fnm{Anna} \sur{Bernasconi}}\email{anna.bernasconi@unipi.it}

\author[1]{\fnm{Alessandro} \sur{Berti}}\email{alessandro.berti@phd.unipi.it}

\author[1]{\fnm{Gianna M.} \sur{Del Corso}}\email{gianna.delcorso@unipi.it}
\author[1]{\fnm{Riccardo} \sur{Guidotti}}\email{riccardo.guidotti@unipi.it}
\author*[1]{\fnm{Alessandro} \sur{Poggiali}}\email{alessandro.poggiali@phd.unipi.it}
\affil*[1]{\orgdiv{Department of Computer Science}, \orgname{University of Pisa}, \orgaddress{\street{Largo Pontecorvo, 3}, \city{Pisa}, \postcode{56127}, \country{Italy}}}

\abstract{
Quantum computing sets the foundation for new ways of designing algorithms, thanks to the peculiar properties inherited by quantum mechanics. The exploration of this new paradigm faces new challenges concerning which field quantum speedup can be achieved. Towards finding solutions, looking for the design of quantum subroutines that are more efficient than their classical counterpart poses solid pillars to new powerful quantum algorithms.
Herewith, we delve into a grounding subroutine, the computation of the variance, whose usefulness spaces across different fields of application, particularly the Artificial Intelligence (AI) one. Indeed, the finding of the quantum counterpart of these building blocks impacts vertically those algorithms that leverage this metric.
In this work, we propose QVAR, a quantum subroutine, to compute the variance that exhibits a logarithmic complexity both in the circuit depth and width, excluding the state preparation cost. With the vision of showing the use of QVAR as a subroutine for new quantum algorithms, we tackle two tasks from the AI domain: Feature Selection and Outlier Detection. In particular, we showcase two AI hybrid quantum algorithms that leverage QVAR: the Hybrid Quantum Feature Selection (HQFS) algorithm and the Quantum Outlier Detection Algorithm (QODA).
In this manuscript, we describe the implementation of QVAR, HQFS, and QODA,  providing their correctness and complexities and showing the effectiveness of these hybrid quantum algorithms with respect to their classical counterpart.}


\maketitle

\section{Introduction}\label{sec1}
The emergence of the Quantum Computation paradigm can be accredited to the prospect of discovering computational speed-up with respect to classical computation.
Quantum Computation introduced novel challenges and new ways to rethink algorithms. Within this emerging computational paradigm, among all the challenges that arose, a pivotal research question deals with the design and implementation of quantum versions of widespread subroutines employed in the classical domain to the quantum one. In particular, identifying and understanding speed-up opportunities for these classical grounding subroutines in their quantum counterpart is crucial for the development of advanced quantum algorithms built upon these foundational components~\cite{nielsen2016quantum,lloyd2013quantum}. 

Among those components, our interest fell on the concept of \textit{variance}, a pervasive measure that finds application across a vast number of domains~\cite{7160458}, including the field of Artificial Intelligence (AI). Although the variance can be computed classically with linear complexity~\cite{chan1983algorithms}, the design of a possibly more efficient quantum algorithm is not straightforward. To the best of our knowledge, there is no prior work on computing the variance over a dataset encoded in the amplitudes.

In this work, we therefore propose a 
Quantum subroutine designed for computing the Variance (QVAR) over a dataset, where the instances are encoded within the amplitudes of a quantum state~\cite{schuld2018supervised}. After state preparation, the subroutine complexity exhibits logarithmic characteristics in both its width and depth, underscoring valuable properties.
The QVAR subroutine lends itself naturally to applications in Artificial Intelligence. Therefore, in this study, we also propose two QVAR-based quantum algorithms: one designed for addressing the  \textit{Feature Selection} tasks (HQFS) and the other tailored for solving the \textit{Outlier Detection} problem (QODA)~\cite{tan2006introduction}. Both algorithms follow a hybrid paradigm where classical and quantum computations are interleaved to solve a given task. This paradigm enables the optimal utilization of current quantum technologies, where problems like short decoherence time, noisy qubits, and limited hardware capabilities make the execution of full quantum algorithms practically impossible.

In this work, we design the QVAR subroutine, and the HQFS, and QODA algorithms. 
The design is supported by defining their respective circuits and presenting mathematical proofs of correctness. The proposed quantum algorithms have been experimentally evaluated, comparing their performance with the classical counterparts. Additionally, we provide the Qiskit\footnote{\url{https://qiskit.org/}} implementations for QVAR\footnote{\url{https://github.com/AlessandroPoggiali/QVAR}}, HQFS\footnote{\url{https://github.com/AlessandroPoggiali/HQFS}}, and QODA\footnote{\url{https://github.com/AlessandroPoggiali/QODA}}. Ultimately, due to the large variety of applications that QVAR can find, we support the development of those paths by creating the PyPI package\footnote{\url{https://pypi.org/project/QVAR/}} of QVAR.
This manuscript represents the extended version of the conference paper presented in~\cite{poggiali2023quantum}. With respect to the conference paper, we give a detailed analytical description of the QVAR subroutine, together with a formal proof of correctness. Moreover, we introduce the novel QODA algorithm for Outlier Detection and enhance the experimental evaluation, providing a more comprehensive exploration of the subject matter.

The rest of the manuscript is organized as follows. 
Section~\ref{sec:related} presents the related works and motivates our solution with respect to the state-of-the-art, 
Section~\ref{sec:qvar} introduces the quantum subroutine for computing the variance (QVAR), Section~\ref{sec:HQFS} presents the Quantum Feature Selection algorithm (HQFS), Section~\ref{sec:QODA} describes the Quantum Outlier Detection algorithm (QODA), 
Section~\ref{sec:experiments} validates the effectiveness of the proposed algorithms, and, eventually, Section~\ref{sec:conclusion} draws the conclusions according to the obtained results.

\section{Related work}\label{sec:related}
The variance measures the variability over a set of data, and it is a fundamental building block to design algorithms whose applications range across different domains such as Statistical Analysis, Optimization Algorithms, and Machine Learning~\cite{7160458}. 

Many algorithms for the computation of the variance with stable and effective approaches have been proposed in the literature (e.g., see~\cite{chan1983algorithms,Ling1974ComparisonOS}). However, as far as we know, the literature on quantum computing lacks specific works on the computation of the variance.  

Hereby, we design a quantum subroutine that computes the variance to address two fundamental problems that leverage the variance itself: the \textit{Feature Selection} and the \textit{Outlier Detection}. 

\textit{Feature Selection} techniques have the objective of selecting the most relevant features from a dataset, a crucial step in data analysis~\cite{tan2006introduction}. Indeed, Feature Selection proves to be notably advantageous in achieving efficient data reduction~\cite{guyon2003introduction,cai2018feature}. 
The subset of features this preprocessing identifies can contribute to developing more robust models.
Among the Feature Selection techniques~\cite{kumar2014feature}, hereby we focus on unsupervised Feature Selection methods, i.e., methods that do not use explicit target labels or class information to select the relevant features. In particular, we consider Feature Selection techniques tailored for numerical variables, wherein features with variances falling below a specified threshold are removed~\cite{tan2006introduction}.
While the existing literature has extensively investigated Feature Selection techniques for classical algorithms, novel challenges emerge in the design of Feature Selection algorithms within the domain of quantum computing. Some Quantum Feature Selection algorithms have been recently proposed in the literature. In~\cite{mucke2022quantum}, the authors propose a feature selection algorithm based on a quadratic unconstrained binary optimization (QUBO) problem. In~\cite{nembrini2021feature}, a solution for the \textit{quantum annealing} paradigm is given. Eventually, the work~\cite{chakraborty2020hybrid} describes a quantum feature selection algorithm that deals with binary values. In this work, we introduce a Feature Selection algorithm (HQFS) based on the gate model, wherein values are encoded within the amplitudes of a quantum state. 

\textit{Outlier Detection} is crucial across various domains~\cite{boukerche2020outlier}, encompassing Fraud Detection, Medical Anomaly Diagnosis, and Engineering. The significance attributed to outliers varies based on the contextual requirements. In data analysis, outliers can significantly impact the results, leading to incorrect conclusions, decisions, and models. 
Conversely, outliers are the critical information crucial to understand and address in domains like fraud detection. Recently, quantum computing has shown great potential to solve computationally expensive problems for classical computers, including data analysis tasks. Some techniques for detecting outliers in quantum computing are present in the literature. In~\cite{Guo_2022,Liang_2019}, the authors present a quantum anomaly detection algorithm based on density estimation. Furthermore, in the study~\cite{mazouzi2020hybrid}, a methodology for anomaly detection is introduced, leveraging both Hamming Distance and Durr-Hoyer algorithms.
With respect to the state-of-the-art, we design a hybrid quantum algorithm for anomaly detection (QODA) inspired by the classical Angle-Based Outlier Detection (ABOD) algorithm~\cite{kriegel2008angle}. ABOD is an unsupervised distance-based method for detecting outliers that measures the abnormality of each data point by computing the variance of the angles between all pairs of vectors in the dataset (see Section~\ref{sec:QODA} for more details). 

\section{A Quantum Subroutine for Variance Estimation}\label{sec:qvar}
In this section, we introduce the QVAR subroutine for computing the variance of a set of values encoded within the amplitudes of a quantum state. 
 
\begin{figure}[t]
\begin{center}
    \centering
    \adjustbox{scale=0.6,center}{

\begin{quantikz}[ row sep=0.2cm, column sep=0.2cm]%
\lstick{$\ket{a}=\ket{0}$}& \qw          & \gate{H} & \ctrl{1} & \qw      & \ctrl{2} & \qw      & \qw \ \ldots\ & \ctrl{4} & \qw      & \ctrl{5} & \ctrl{6} & \qw \ \ldots\ & \ctrl{8} & \ctrl{9} & \ctrl{10}& \qw \ \ldots\ & \ctrl{12} & \ctrl{1} & \ctrl{2} & \qw \ \ldots\ & \ctrl{4} & \gate{H} & \qw      & \qw       \\
\lstick{$\ket{e_1}=\ket{0}$ }& \qw       & \qw      & \targ{}  & \gate{X} & \qw      & \qw      & \qw \ \ldots\ & \qw      & \qw      & \qw      & \qw      & \qw \ \ldots\ & \qw      & \qw      & \qw      & \qw \ \ldots\ & \qw       & \gate{H} & \qw      & \qw \ \ldots\ & \qw      & \qw      & \qw      & \qw       \\
\lstick{$\ket{e_2}=\ket{0}$}& \qw        & \qw      & \qw      & \qw      & \targ{}  & \gate{X} & \qw \ \ldots\ & \qw      & \qw      & \qw      & \qw      & \qw \ \ldots\ & \qw      & \qw      & \qw      & \qw \ \ldots\ & \qw       & \qw      & \gate{H} & \qw \ \ldots\ & \qw      & \qw      & \qw      & \qw       \\
\lstick{$\ldots$}               &\\
\lstick{$\ket{e_n}=\ket{0}$}& \qw        & \qw      & \qw      & \qw      & \qw      & \qw      & \qw \ \ldots\ & \targ{}  & \gate{X} & \qw      & \qw      & \qw \ \ldots\ & \qw      & \qw      & \qw      & \qw \ \ldots\ & \qw      & \qw       &  \qw     & \qw \ \ldots\ & \gate{H} & \qw      & \qw      & \qw       \\

\lstick{$\ket{q_1}=\ket{0}$}& \qw        & \qw      & \qw      & \qw      & \qw      & \qw      & \qw \ \ldots\ & \qw      & \qw      & \swap{4} & \qw      & \qw \ \ldots\ & \qw      & \qw      & \qw      & \qw \ \ldots\ & \qw      & \qw       & \qw      & \qw \ \ldots\ & \qw      & \gate{H} & \gate{X} & \qw       \\
\lstick{$\ket{q_2}=\ket{0}$} & \qw        & \qw      & \qw      & \qw      & \qw      & \qw      & \qw \ \ldots\ & \qw      & \qw      & \qw      & \swap{4} & \qw \ \ldots\ & \qw      & \qw      & \qw      & \qw \ \ldots\ & \qw      & \qw       & \qw      & \qw \ \ldots\ & \qw      & \gate{H} & \gate{X} & \qw       \\
\lstick{$\ldots$}               & \\
\lstick{$\ket{q_n}=\ket{0}$}  & \qw       & \qw      & \qw      & \qw      & \qw      & \qw      & \qw \ \ldots\ & \qw      & \qw      & \qw      & \qw      & \qw \ \ldots\ & \swap{4} & \qw      & \qw      & \qw \ \ldots\ & \qw      & \qw       & \qw      & \qw \ \ldots\ & \qw      & \gate{H} & \gate{X} & \qw       \\ 

\lstick{$\ket{i_1}=\ket{0}$} & \gate[wires=4,nwires={3}]{\begin{array}{c}\text{Amplitude} \\ \text{Encoding}\\ \text{\ket{D}}\end{array}}  & \qw      & \qw      & \qw      & \qw      & \qw      & \qw \ \ldots\ & \qw      & \qw      & \swap{}  & \qw      & \qw \ \ldots\ & \qw      & \gate{H} & \qw      & \qw \ \ldots\ & \qw      & \qw       & \qw      & \qw \ \ldots\ & \qw      & \qw      & \qw      & \qw        \\
\lstick{$\ket{i_2}=\ket{0}$} & \qw        & \qw      & \qw      & \qw      & \qw      & \qw      & \qw \ \ldots\ & \qw      & \qw      & \qw      & \swap{}  & \qw \ \ldots\ & \qw      & \qw      & \gate{H} & \qw \ \ldots\ & \qw      & \qw       & \qw      & \qw \ \ldots\ & \qw      & \qw      & \qw      & \qw        \\
\lstick{$\ldots$}               & \\
\lstick{$\ket{i_n}=\ket{0}$} & \qw        & \qw      & \qw      & \qw      & \qw      & \qw      & \qw \ \ldots\ & \qw      & \qw      & \qw      & \qw      & \qw \ \ldots\ & \swap{}  & \qw      & \qw      & \qw \ \ldots\ & \gate{H} & \qw       & \qw      & \qw \ \ldots\ & \qw      & \qw      & \qw      & \qw       \\
\end{quantikz}}
\caption{QVAR Oracle}
\label{fig:qvar}
\end{center}
\end{figure}
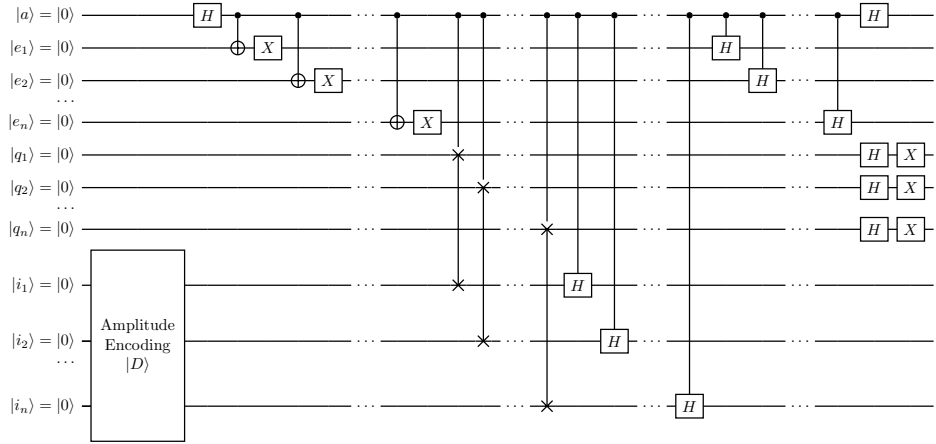

Given a set $D = \{d_0, d_1, \dots, d_{N-1}\}$ of $N$ values, we are interested in estimating the variance of $D$, defined as
$$
\rm{Var}(D)=\frac{1}{N} \sum_{t=0}^{N-1} (d_t-\mathbb{E}[D])^2,
$$
where $\mathbb{E}[D]=\frac{1}{N}\sum_{t=0}^{N-1} d_t$ is the mean value of the data. 

Note that $D$ can be either a set of classical values or the result of a previous quantum computation, i.e., a superposition of values. In the first case, one should adopt an encoding algorithm for loading the classical data into the quantum state. However, QVAR can generally be applied as a quantum subroutine to a quantum superposition of values indexed by proper register qubits. 

The QVAR subroutine comprises the quantum oracle reported in Figure \ref{fig:qvar}, which the Amplitude Estimation algorithm~\cite{brassard2002quantum} uses to estimate the variance value. This algorithm requires in total $3n+s+1$ qubits, where $n=\ceil{\log_2 N}$ and $s=O(\log\frac{1}{\varepsilon})$ is the number of additional qubits required by Amplitude Estimation to get an estimate of the variance with (absolute) error $\varepsilon$ with respect to the classical variance. 
The formal proof of correctness of the QVAR oracle of Figure~\ref{fig:qvar} is reported below. The main idea is to use a qubit ancilla $\ket{a}$ to create an equal superposition of two branches through a Hadamard gate. In the branch where the ancilla is in state $\ket{0}$, we maintain the superposition encoding the values in $D$, while in the branch where the ancilla is in state $\ket{1}$, we apply a sequence of gates that compute the mean $\mathbb{E}[D]$. Eventually, we apply another Hadamard gate on $\ket{a}$ to make the two branches collide and create the superposition containing the component of the variance, namely the differences $d_t - \mathbb{E}[D]$, for $t=0,\dots,N-1$. 
Since all these differences are the amplitudes of the state vector related to specific configurations of the qubits, we can use the Amplitude Estimation algorithm to estimate the variance as the sum of squares of the target amplitudes. In our case, we employ AE to estimate the amplitude of the target configuration $\ket{ae^{\otimes n}q^{\otimes n}} = \ket{1 1^{\otimes n} 0^{\otimes n} }$ by measuring the $s$ additional qubits. The measurement output represents an approximation of the variance in the computational basis. Concerning the complexity of this method, in general, the depth of AE is $O(\delta\frac{1}{\varepsilon}+\log \log \frac{1}{\varepsilon})$~\cite{giurgica2022low} where $\delta$ is the depth of the oracle. In our case, assuming the existence of an efficient quantum state preparation technique, $\delta=O({\rm{polylog}} \, N)$. Thus, the overall complexity of QVAR is $O(\frac{1}{\varepsilon}{\rm{polylog}} \, N)$. We can increase the precision of QVAR by interpolating the measurement probabilities and computing the maximum likelihood estimator (we call this method ML-QVAR)~\cite{grinko2021iterative}.

The correctness of QVAR is proved in the  following proposition.

\begin{proposition}
Let  $D = \{d_0, d_1, \dots, d_{N-1}\}$ be a set of $N$ values encoded in the quantum state of $n=\log N$ qubits as follows:
\begin{equation*}
\ket{D}=\frac{1}{\sqrt{N}} \sum_{t=0}^{N-1}d_t \ket{t}_i,
\end{equation*}
where $\sqrt{N}$ is the normalization factor. Then, the quantum oracle of the QVAR algorithm 
starting from the quantum state
\begin{equation*}
    \ket{\psi_0} = \ket{0}_a\ket{0}_e^{\otimes n}\ket{0}_q^{\otimes n} \ket{D},
\end{equation*}
returns in the configuration of the state vector equal to  $\ket{1}_a\ket{1}_e^{\otimes n} \ket{0}_q^{\otimes n}$ the terms $d_t-\mathbb{E}[D]$. 
\end{proposition}
\begin{proof}
Following the quantum operations depicted in Figure~\ref{fig:qvar}, the Hadamard gate on $\ket{0}_a$ creates two branches of computation: the branch where the qubit is in state $\ket{0}$ at the end of the computation will contain the superposition of all $d_t$. While the branch where the qubit is in state $\ket{1}$ will contain $N$ copies of the mean $\mathbb{E}[D]$. Indeed, the first $H$ gate on the ancilla qubit evolves the initial state $\ket{\psi_0}$ in 
{\small\begin{equation*}
    \ket{\psi_1} = \frac{1}{\sqrt{2N}}\left(\ket{0}_a \ket{0}_e^{\otimes n}\ket{0}_q^{\otimes n}\sum_{t=0}^{N-1}d_t\ket{t}_i +\ket{1}_a \ket{0}_e^{\otimes n}\ket{0}_q^{\otimes n}\sum_{t=0}^{N-1}d_t\ket{t}_i\right).
\end{equation*}}
The C-NOT, controlled by the ancilla qubit, and the X gates on register $e$ bring the state in:
 \begin{equation*}
    \ket{\psi_2} = \frac{1}{\sqrt{2N}}\left(\ket{0}_a \ket{1}_e^{\otimes n}\ket{0}_q^{\otimes n}\sum_{t=0}^{N-1}d_t\ket{t}_i +\ket{1}_a \ket{0}_e^{\otimes n}\ket{0}_q^{\otimes n}\sum_{t=0}^{N-1}d_t\ket{t}_i\right).
\end{equation*}
The C-SWAP gates swap the qubits of register $i$ and register $q$ in the branch where the ancilla qubit is in state $\ket{1}$, leading to the state:
\begin{equation*}
    \ket{\psi_3} = \frac{1}{\sqrt{2N}}\left(\ket{0}_a \ket{1}_e^{\otimes n}\ket{0}_q^{\otimes n}\sum_{t=0}^{N-1}d_t\ket{t}_i +\ket{1}_a \ket{0}_e^{\otimes n}\sum_{t=0}^{N-1}d_t\ket{t}_q\ket{0}_i^{\otimes n}\right).
\end{equation*}
Since we want the mean of $D$ in the branch where the ancilla is \ket{1}, we apply $H$ gates on each qubit of register $i$ controlled by the value of the ancilla. This produces the following state:
{\begin{equation*}
    \ket{\psi_4} = \frac{1}{\sqrt{2N}}\left(\ket{0}_a \ket{1}_e^{\otimes n}\ket{0}_q^{\otimes n}\sum_{t=0}^{N-1}d_t\ket{t}_i +\ket{1}_a \ket{0}_e^{\otimes n}\sum_{t=0}^{N-1}d_t\ket{t}_q \ket{+}_i^{\otimes n}\right),
\end{equation*}}
where $\ket{+}=\frac{\ket{0}+\ket{1}}{\sqrt{2}}.$
In order to have the mean value $\mathbb{E}[D]$, we apply an H gate on each qubit of register $e$ controlled by the ancilla. This results in:
\begin{equation*}
    \ket{\psi_5} = \frac{1}{\sqrt{2N}}\left(\ket{0}_a \ket{1}_e^{\otimes n}\ket{0}_q^{\otimes n}\sum_{t=0}^{N-1}d_t\ket{t}_i +\ket{1}_a \ket{+}_e^{\otimes n}\sum_{t=0}^{N-1}d_t\ket{t}_q \ket{+}_i^{\otimes n}\right).
\end{equation*}
The next step involves applying another $H$ gate on the first qubit, which will later be exploited to collide each $d_t$ with the mean $\mathbb{E}[D]$. The state obtained is:
\begin{equation*}
\begin{split}
    \ket{\psi_6} &= \frac{1}{2\sqrt{N}}\left(\ket{0}_a \ket{1}_e^{\otimes n}\ket{0}_q^{\otimes n}\sum_{t=0}^{N-1}d_t\ket{t}_i +\ket{1}_a\ket{1}_e^{\otimes n}\ket{0}_q^{\otimes n}\sum_{t=0}^{N-1}d_t\ket{t}_i\right. \\ & 
   \left. + \ket{0}_a \ket{+}_e^{\otimes n}\sum_{t=0}^{N-1}d_t\ket{t}_q \ket{+}_i^{\otimes n} -\ket{1}_a \ket{+}_e^{\otimes n}\sum_{t=0}^{N-1}d_t\ket{t}_q \ket{+}_i^{\otimes n} \right).
\end{split}
\end{equation*}
Then, the $H$ gates on each qubit of register $q$ are required to form the actual sums between each $d_t$. 
Recall that $H^{\otimes n} \ket{t}^{\otimes n} = \tfrac{1}{\sqrt{N}} \sum_{z=0}^{N-1}(-1)^{t\cdot z}\ket{z}$, where $t \cdot z$ is the scalar product between the binary representations of $t$ and $z$. The state becomes:
\begin{equation*}
\begin{split}
    \ket{\psi_7} &= \frac{1}{2\sqrt{N}}\left(\ket{0}_a\ket{1}_e^{\otimes n}\ket{+}_q^{\otimes n}\sum_{t=0}^{N-1}d_t\ket{t}_i +\ket{1}_a \ket{1}_e^{\otimes n}\ket{+}_q^{\otimes n}\sum_{t=0}^{N-1}d_t\ket{t}_i  \right.\\ & 
    +\ket{0}_a \ket{+}_e^{\otimes n}\sum_{t=0}^{N-1}d_t\left(\frac{1}{\sqrt{N}}\sum_{z=0}^{N-1}(-1)^{t\cdot z}\ket{z}_q\right) \ket{+}_i^{\otimes n} \\& - 
   \left.   \ket{1}_a \ket{+}_e^{\otimes n}\sum_{t=0}^{N-1}d_t \left(\frac{1}{\sqrt{N}}\sum_{z=0}^{N-1}(-1)^{t\cdot z}\ket{z}_q\right) \ket{+}_i^{\otimes n}\right)\\ &
     =  \frac{1}{2\sqrt{N}}\left(\ket{0}_a\ket{1}_e^{\otimes n}\ket{+}_q^{\otimes n}\sum_{t=0}^{N-1}d_t\ket{t}_i +\ket{1}_a \ket{1}_e^{\otimes n}\ket{+}_q^{\otimes n}\sum_{t=0}^{N-1}d_t\ket{t}_i \right.\\&
        +\frac{1}{\sqrt{N}} \ket{0}_a \ket{+}_e^{\otimes n}\sum_{t=0}^{N-1}d_t\left(\ket{0}_q^{\otimes n}
    + \sum_{z=1}^{N-1}(-1)^{t\cdot z}\ket{z}_q\right) \ket{+}_i^{\otimes n}  \\
    &- \left.\frac{1}{\sqrt{N}} \ket{1}_a \ket{+}_e^{\otimes n}\sum_{t=0}^{N-1}d_t \left(\ket{0}_q^{\otimes n} + \sum_{z=1}^{N-1}(-1)^{t\cdot z}\ket{z}_q\right) \ket{+}_i^{\otimes n}\right).
\end{split}
\end{equation*}
We are interested in the configuration $\ket{1}_a\ket{1}_e^{\otimes n} \ket{0}_q^{\otimes n}$: 
\begin{equation*}
\begin{split}
\bra{1}_a\bra{1}_e^{\otimes n} \bra{0}_q^{\otimes n}\ket{\psi_7}=&\frac{1}{2\sqrt{N}}  \left(\frac{1}{(\sqrt{N}}\sum_{t=0}^{N-1}d_t\ket{t}_i-\frac{1}{2^n}\sum_{t=0}^{N-1}d_t\ket{+}_i^{\otimes n}\right)\\&
= \frac{1}{2\sqrt{N}}\left(\frac{1}{\sqrt{N}}\sum_{t=0}^{N-1}d_t \ket{t}_i -\frac{1}{N}\ket{+}_i^{\otimes n}\sum_{t=0}^{N-1}d_t\right)\\&
\end{split}
\end{equation*}
Changing the index in the last sum, and expressing $\ket{+}_i^{\otimes n}$ in the computational basis, we obtain:
\begin{equation*}
\begin{split}
&\frac{1}{2N}\left(\sum_{t=0}^{N-1}d_t \ket{t}_i -\sum_{t=0}^{N-1}\ket{t}_i\frac{1}{N}\sum_{k=0}^{N-1}d_k\right) \\&
= \frac{1}{2N}\ket{1}_a\ket{1}_e^{\otimes n}\ket{0}_q^{\otimes n}\sum_{t=0}^{N-1} \ket{t}_i \left(d_t-\frac{1}{N}\sum_{k=0}^{N-1}d_k\right),
\end{split}
\end{equation*}
from which it is evident that the components of the variance, up to a normalization factor, are associated with the configurations of the state vector equal to $\ket{1}_a\ket{ 1}_e^{\otimes n} \ket{0}_q^{\otimes n}$.
\end{proof}

\section{Quantum Feature Selection}\label{sec:HQFS}
In this section, we study the application of the QVAR subroutine to the Feature Selection task. As discussed in~\cite{tan2006introduction}, given a multidimensional dataset,  we can filter out features with low variance. Indeed, low variance features do not vary significantly among the different records in the dataset and are less useful for discriminative tasks.   
Hereby, we design an efficient Hybrid Quantum Feature Selection (HQFS) algorithm that exploits the QVAR subroutine for discarding low variance features.  A detailed description follows.

Given a multidimensional dataset $D \in \mathbb{R}^{M \times N}$ with $M$ records and $N$ features, we denote as $D_{i,j}$ the $j$-th feature value of the $i$-th record. The HQFS algorithm iterates over the features and computes the variances $\rm{Var}_j =\rm{Var}(d_{0,j},\cdots,d_{M-1,j})$ for $j \in [ 0,\cdots,N-1]$. HQFS employs a distinct QVAR circuit for each feature $j$ to calculate each variance. Overall, HQFS uses $N$ QVAR circuits, one for each feature. Then, HQFS drops all features whose variance is below a given threshold $t$.

\begin{algorithm2e}[ht!]
\scriptsize
        \KwIn{$D$ - input dataset, $t$ - variance threshold}
        \KwOut{$F$ - list of selected features}
        \BlankLine
        $F \gets \{f_0,\cdots,f_{N-1}\}$ \tcp*[f]{set of selected features initially containing all features}
        \BlankLine
        \For(\tcp*[f]{classically iterate among all feature}){$j \in [0,\cdots N-1]$}{
                $\rm{Var}_j\gets QVAR(D_{:,j})$ \tcp*{compute the variance using the QVAR algorithm.}
                \BlankLine
                \If(\tcp*[f]{check if feature $f_j$ is uninformative}){$\rm{Var}_j \leq t$}{
                 $F \gets F \setminus \{f_j\}$
            \tcp*{remove feature $f_j$ to the set of selected features}
            }
            }
        \Return $F$ \tcp*{return the list of selected features}
\caption{HQFS}
\label{alg:HQFS}
\end{algorithm2e}

We can use either the QVAR or ML-QVAR as a variance estimation method. In the following, we refer to HQFS as the version using QVAR and ML-HQFS as the version using ML-QVAR. Since the depth of AE is $O(\delta\frac{1}{\varepsilon}+\log \log \frac{1}{\varepsilon})$, with $\delta$ the depth of the oracle, which in our case is $O({\rm{polylog}} \,M)$, the overall complexity of Algorithm~\ref{alg:HQFS} is $O(N(\frac{1}{\varepsilon}{\rm{polylog}} \,M))$, if we assume available efficient methods for reconstructing the initial state. 

\section{Quantum Outlier Detection}\label{sec:QODA}
This section introduces a hybrid quantum algorithm inspired by the classical Angle-Based Outlier Detection (ABOD) algorithm~\cite{kriegel2008angle}. 

ABOD is a distance-based unsupervised Outlier Detection method that measures the abnormality of each data point by computing the variance of the angles between the difference vectors of the given point and the other points in the dataset. Indeed, the variance of these angles depends on how different a point is from the others:  for points within a cluster, the angles differ widely, while outliers are expected to present a smaller variance because they are further away from the rest of the dataset. Hence, the points with smaller ABOD values are considered to be more likely to be outliers.

ABOD has several advantages over other Outlier Detection techniques, including detecting outliers in high-dimensional datasets. However, it may not perform well on datasets with a high degree of overlap between classes or when the dataset contains clusters of outliers. Moreover, it has a cost of $O(M^3 + M^2 N)$ where $M$ is the number of records, and $N$ is the dimension of each record; hence the cost is very high. 

The classical algorithm works as follows. 
Assume we have $M$ records, each with $N$-dimensions, stored in a matrix $X=[x_1| x_2|\cdots|x_M] \in \mathbb{R}^{N\times M}$, so that each record can be seen as belonging to a vector space of dimension $N$. We iterate over all records, considering each one as a pivot. Thus, let $x_p$, $1 \le p \le M$, be the current pivot.  We denote by $\theta_{ij}^{(p)}$ the angle between the vectors  $x_i$ and $x_j$ observed from the pivot $x_p$, i.e., the angle between the difference vectors  $\tilde x_i = x_i - x_p$ and $\tilde x_j = x_j - x_p$. 
Mathematically, we can write this as:
\begin{equation}\label{eq:theta}
    \theta_{ij}^{(p)}=\arccos\left(\frac{\tilde x_i^T \tilde x_j}{\|\tilde x_i\|\|\tilde x_j\|}\right),
\end{equation}
where all norms are 2-norms.
For each $p\in \{ 1,   \ldots, M\}$, we compute  the variance $v_p$ of $\theta_{ij}^{(p)}$, with $1 \le i < j \le M$ and $i, j \neq p$, i.e., 
$v_p=\mbox{Var}\left(\theta_{ij}^{(p)}\right)=\mathbb{E}{[(\theta_{ij}^{(p)})^2]}-\mathbb{E}^2[\theta_{ij}^{(p)}]$. 
 
Once the vector $V=[v_1,v_2, \ldots, v_M]$  containing the variances is ready, the algorithm detects the outliers as the records with variance below a given threshold. The algorithm is usually implemented by computing the variance directly on the inner products, i.e., on the cosine of $\theta_{ij}^{(p)}$~\cite{kriegel2008angle}.

\begin{figure}[t]

\begin{quantikz}[transparent, row sep=0.05cm, column sep=0.3cm]%
\lstick{$\ket{d}=\ket{0}$}             &\qwbundle{1} & \gate{H}      & \ctrl{2} & \gate{X}                        & \ctrl{1} & \qw           &  \gate{H}                        &   \qw              \\
\lstick{$\ket{j}=\ket{0}^{\otimes m}$} &\qwbundle{m} & \gate{H} & \qw & \qw &\gate[wires=3, label style={yshift=0.8cm}]{\begin{array}{c}\text{Amplitude} \\ \text{Encoding}\\ \text{\ket{D}}\end{array}} \qw & \qw & \qw & \qw & \\
\lstick{$\ket{i}=\ket{0}^{\otimes m}$} &\qwbundle{m}& \gate{H} \qw & \gate[wires=2]{\begin{array}{c}\text{Amplitude} \\ \text{Encoding}\\ \text{\ket{D}}\end{array}} \qw & \qw & \linethrough\qw & \qw & \qw & \qw & \\
\lstick{$\ket{k}=\ket{0}^{\otimes n}$} &\qwbundle{n}& \gate{H} \qw & \qw & \qw & \qw & \qw & \qw & \qw &\\
\end{quantikz}

\caption{Computation of differences.}
\label{fig:differences_scomputation}
\end{figure}
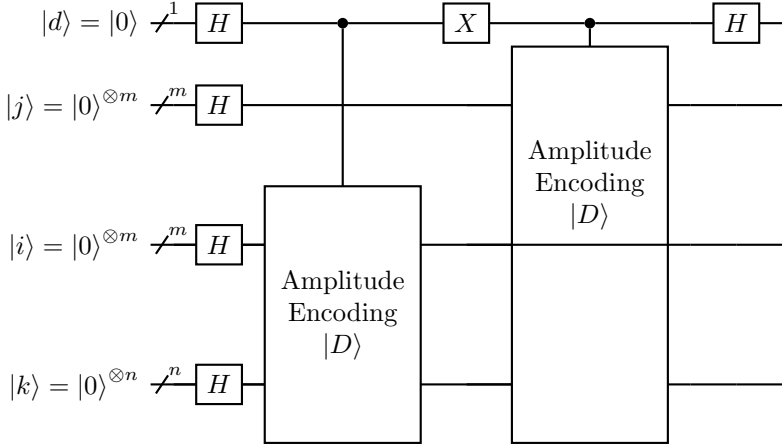

We design a hybrid quantum version of ABOD called Quantum Outlier Detection Algorithm (QODA). QODA iterates over all records in a given dataset and calls a quantum subroutine to test whether the current record is an outlier on a single quantum circuit. Alternatively, we could test each record in the dataset by employing distinct quantum circuits that can run in parallel.  

QODA mimics the classical ADOB approach, exploiting the variance as a statistical measure to detect
outliers, allowing a reduction in the computational complexity of the classical version thanks to the quantum oracle for the computation of the variance proposed in Section~\ref{sec:qvar}. 
A crucial subroutine of QODA is, therefore, an efficient quantum heuristic for estimating the variance of all angles  $\theta_{ij}^{(p)}$, for a given sample $x_p$, taken as a pivot. Thus, we employ the QVAR subroutine defined in Section~\ref{sec:qvar}.

This heuristic is based on the idea of approximating the variance of the angles $\theta_{ij}^{(p)}$ in~\eqref{eq:theta} with the variance of the differences  between the vectors translated by the pivot $x_p$ and then normalized. 
For all $1 \le i \le M$,  we (classically) compute $\tilde x_i = x_i - x_p$. Then, we normalize all vectors  so that the denominators in equation~\eqref{eq:theta} vanish, adopting as default data preprocessing the Inverse Stereographic Projection
(ISP), in order to keep the clusters separated, as discussed in~\cite{poggiali2022quantum}. 

Recall that the ISP maps $N$-dimensional data into the surface of a unit sphere
in the $(N+1)$-dimensional space and is computed as follows.
Given a vector $v=(v_1, v_2, \ldots, v_N) \in \mathbb{R}^N$, we have: 

$$\mbox{ISP}(v)=\left(\frac{2v_1}{\|v\|^2+1},\frac{2v_2}{1+\|v\|^2}, \ldots, \frac{2v_N}{1+\|v\|^2}, \frac{-1+\|v\|^2}{1+\|v\|^2}\right).
$$ 
Notice that  $\|\mbox{ISP}(v)\| = 1$. 
Thus, let $\hat x_i = \mbox{ISP}(\tilde x_i) = \mbox{ISP}(x_i - x_p) $ denote the normalized vectors, translated by $x_p$.  
We denote the scaled components of the differences between two projected records $\hat x_i$ and $\hat x_j$ as:
\begin{equation} \label{eq:delta}
\Delta^{(p)}_{ijk} = \frac{1}{2(M-1)\sqrt{N+1}} \left(\hat x_{i} - \hat x_{j}\right)_k, \quad {i, j=0, \ldots, M-1, \, k=0, \ldots, N-1.}
\end{equation} 
After encoding all the vectors $\hat x_i$ in a quantum circuit, for example using one of the QRAM models~\cite{park2019circuit,kerenidis2016quantum}, we can use  a Hadamard gate to compute simultaneously all the differences $\Delta^{(p)}_{ijk}$ in ~\eqref{eq:delta} in the amplitudes, appropriately normalized. In~Figure~\ref{fig:differences_scomputation}, we show the circuit used for this purpose. The state $$
\ket{D}= \frac{1}{\sqrt{MN}} \sum_{i=0}^{M-1}\sum_{k=0}^{N-1}\hat{x}_{ki} \ket{i}\ket{k}$$
corresponding to the input dataset $D$ is obtained using a QRAM model, and is uploaded twice in the circuit, in two different branches of the qubit $\ket{d}$, so that it is possible to compute all the differences $\Delta_{ijk}^{(p)}$ with just a Hadamard gate on the first qubit $\ket{d}$. Indeed, the state vector after the final Hadamard gate contains all the difference $(\hat x_j-\hat x_i)_k$.

Note that the dependence from the chosen pivot $x_p$ is still within $\hat x_i$ and $\hat x_j$, indeed we have 
$$
(\hat x_i - \hat x_j)_k= \left\{ \begin{array}{ll}
\frac{2(x_{ki}-x_{kp})}{1+\|x_i-x_p\|^2}-\frac{2(x_{kj}-x_{kp})}{1+\|x_j-x_p\|^2} & \mbox{ for } k=1, \ldots, N\\
& \\
\frac{2(\|x_i-x_p\|^2-\|x_j-x_p\|^2)}{(1+\|x_i-x_p\|^2)(1+\|x_j-x_p\|^2)},& \mbox{ for } k=N+1.
\end{array}
\right.
$$

In Proposition~\ref{prop2} we show that the variance of the differences $\Delta_{ijk}^{(p)}$ is a good approximation of the variance of the angles, and thus we can apply the circuit for calculating the variance in Figure~\ref{fig:qvar} to the differencies $\Delta_{ijk}^{(p)}$ instead of the angles $\theta_{ij}^{(p)}$. Indeed, the computation of the angles $\theta_{ij}^{(p)}$ requires the estimation of the inner products between the vectors and is more expensive than just applying a Hadamard gate to get the $\Delta_{ijk}^{(p)}$ in superposition.

\smallskip
\begin{proposition}\label{prop2}
Let $\Delta^{(p)}$ be the $(M-1)^2(N+1)$-dimensional vector
obtained appending all the $\Delta_{ijk}^{(p)}$, for $0 \le i, j \le M-1, (i,j \neq p), 0\le k\le N,$~and let $\theta^{(p)}$ be the $(M-1)^2$-dimensional vector whose components are the angles ${\theta}_{ij}^{(p)}$ between $\hat x_i$ and $\hat x_j$, for $1 \le i, j \le M, (i,j \neq p)$. Then:

\begin{equation}
    \rm{Var}(\Delta^{(p)}) \ < \ \frac{1}{N+1} \rm{Var}(\theta^{(p)})+\frac{1}{N+1}\mathbb{E}^2[\theta^{(p)}].
    \label{teo}
\end{equation}
\end{proposition}
\begin{proof} 
First of all, note that $\Delta^{(p)}$ is a $(M-1)^2(N+1)$-dimensional vector
since after the ISP procedure the records have $N+1$ entries, and each pivot is compared with the remaining $M-1$. Then, from the definition, we have
$$
\mbox{Var}(\Delta^{(p)})=\frac{1}{(M-1)^2(N+1)} \sum_{i, j \neq p}\sum_k (\hat{x}_i-\hat{x}_j)_k^2.
$$
Expanding the product, using the fact that each record is normalized and that $\cos(\theta_{ij}^{(p)})=\sum_{k} \hat {x}_{ki}\hat {x}_{kj}$ we have
$$
\mbox{Var}(\Delta^{(p)})=\frac{2}{(M-1)^2(N+1)}((M-1)^2-\sum_{i, j \neq p} \cos(\theta_{ij}^{(p)})).
$$
We can now use the Taylor expansion at zero of the cosine function, truncated at the second term, to get that for each pair $i, j\neq p$ 
$$
\cos(\theta_{ij}^{(p)})\ge 1-\frac{1}{2} (\theta_{ij}^{(p)})^2.
$$
Then, we can upper bound the variance as
\begin{equation} \label{varineq}
\mbox{Var}(\Delta^{(p)})\le \frac{1}{N+1}\mathbb{E}[(\theta^{(p)})^2].
\end{equation}
Since
$\mathbb{E}[(\theta^{(p)})^2]=\mbox{Var}(\theta^{(p)})+\mathbb{E}^2[\theta^{(p)}],$
we can replace $\mathbb{E}[(\theta^{(p)})^2]$ in~\eqref{varineq}  and we get
$$
\mbox{Var}(\Delta^{(p)})\le \frac{1}{N+1} \left(\mbox{Var}(\theta^{(p)}) +\mathbb{E}^2[\theta^{(p)}] \right).
$$   
\end{proof}

\begin{algorithm2e}[t]
\scriptsize
        \KwIn{$X$ - input data, $t$ - threshold}
        \KwOut{$O$ - set of outliers}
        \BlankLine
        $O \gets \emptyset$ \tcp*[f]{set of outliers initially empty}
        \BlankLine
        \For(\tcp*[f]{for each record}){$p\in \{ 1, 2, \ldots, M\}$}{
                $\tilde X \gets X - \ket{x_p}\bra{\bf{1}}$ \tcp*{scale the dataset w.r.t. $x_p$}
                $\hat X \gets \mathit{ISP}(\tilde X)$ \tcp*{normalize $X_p$ with ISP} 
                $v_p \gets \mathit{QVAR}(\hat X)$ \tcp*[f]{Quantum Outlier Factor of $x_p$.}
                \BlankLine
                \If(\tcp*[f]{check if $x_p$ is an outlier}){$v_p \leq t$}{
                 $O \gets O \cup \{x_p\}$
            \tcp*{add $x_p$ to the set of outliers}
            }
               
            }
        \Return $O$ \tcp*{return the list of outliers}
\caption{QODA}
\label{alg:QODAgenerale}
\end{algorithm2e}

Using vector differences instead of angles is convenient as we can maintain the superposition of all differences in the QVAR oracle and perform the variance computation directly on them within the same circuit (see Figure~\ref{fig:qvar}). 
Thus, this approach requires only the data encoding step before applying the quantum oracle QVAR and the final measurement. 
 
The complexity of the quantum oracle QVAR is, in this case, ${\rm{polylog}}\,M^2N$, i.e., ${\rm{polylog}}\,MN$, assuming available efficient methods for encoding classical data, while the number of qubits required is $7m+4n+s+1$, where $m=\ceil{\log_2 M}$, $n=\ceil{\log_2 N}$, and $s$ is the number of additional qubits of AE. 
The overall hybrid algorithm QODA to solve the Outlier Detection problem is summarized by the pseudocode in Algorithm~\ref{alg:QODAgenerale}. The complexity of QODA is $O(M^2N)$ classical operations for data preprocessing plus a $O(M(\frac{1}{\varepsilon}{{\rm{polylog}} MN}))$ quantum cost in depth.

\section{Experiments}\label{sec:experiments}
In this section, we report the experiments to evaluate the QVAR subroutine\footnote{The Qiskit implementation of the QVAR subroutine is publicly available at \url{https://github.com/AlessandroPoggiali/QVAR}} applied to both tasks proposed in the previous sections. We implemented the HQFS algorithm and the QODA using IBM's Qiskit framework\footnote{HQFS code publicly available at \url{https://github.com/AlessandroPoggiali/HQFS}}\footnote{QODA code publicly available at \url{https://github.com/AlessandroPoggiali/QODA}}. 
The main aim of these experiments is to assess the effectiveness of the quantum algorithms HQFS and QODA with respect to their classical counterparts.
Due to technological limits in assessing quantum algorithms on real quantum hardware on meaningful datasets, we conduct the experiments using the \textsc{qasm simulator} of Qiskit, which simulates quantum circuits using classical hardware. In Section~\ref{sec:experiment_hqfs}, we evaluate the performance of the HQFS algorithm on two synthetic datasets and two real datasets. In Section~\ref{sec:experiment_qoda}, we assess the capabilities of QODA in finding outliers over four synthetic datasets and three real datasets. 

\subsubsection*{Validation metrics}
For the HQFS evaluation, we are interested in comparing the ranking of the features returned by our hybrid algorithm over the classical ranking. To measure the similarity between the rankings, we use the Rank Biased Overlap (RBO) measure~\cite{webber2010similarity}, which weights the ranking similarity on the top ranks and returns a value in $[0,1]$. A high value of RBO corresponds to a high similarity on the top ranks. 

As a measure of the effectiveness of our algorithms HQFS and ML-HQFS, we consider the accuracy:
$$
\rm{ACC} =\frac{ \# TP + \# TN}{N},
$$
where $\# \rm{TP}$ and $\# \rm{TN}$ are the number of features correctly discarded/non-discarded, respectively.

Concerning Outlier Detection, for evaluating the tightness of the lower bound of Proposition~\ref{prop2}, we use the following metrics:
\begin{itemize}
    \item Mean Absolute Error (MAE):
    $$
    \rm{MAE} = \frac{1}{M} \sum_{p=0}^{M-1}| a_p-b_p|
    $$
    \item Mean Squared Error (MSE):
    $$
    \rm{MSE}= \frac{1}{M} \sum_{p=0}^{M-1}| a_p-b_p|^2
    $$
    \item Root Mean Squared Error (RMSE):
    $$
    \rm{RMSE}= \sqrt{\rm{MSE}}
    $$
\end{itemize}
where $a_p$ is the right hand side of equation~\eqref{teo}, and $b_p$ is the left hand side of the inequality.

Similarly to the Feature Selection case, we are interested in comparing the ranking of outliers provided by our hybrid algorithm against the one returned by the classical algorithm. In addition to the RBO measure, we evaluate QODA using the Precision at $n$ (P@$n$), by determining the number of records in the top $n$ ranks identified as outliers using differences and angles, respectively.

\subsection{Experiments for Feature Selection}
\label{sec:experiment_hqfs}
We first evaluate the performance of HQFS and ML-HQFS on two synthetic datasets by varying the parameter $s=O(\log\frac{1}{\varepsilon})$, the number of additional qubits required by Amplitude Estimation to get an estimate of the variance with (absolute) error $\varepsilon$ with respect
to the classical variance. Then, we test HQFS and ML-HQFS with $s=6$ on two real datasets. We assess our results by comparing the similarities of the final ranking of features of HQFS and ML-HQFS with respect to the ranking from the classical variance.

\subsubsection*{Dataset} 
For the experimental evaluation of HQFS, we consider two synthetic datasets (\texttt{synth\_1} and \texttt{synth\_2}) and the well-known real datasets \texttt{wine} and \texttt{breast}. Both synthetic datasets have $N=32$ records and $M=10$ features: 7 informative features with high variance and 3 uninformative features with low variance, which are not expected to provide useful information for the analysis. The informative features are sampled from uniform distributions in [-1,1], while the uninformative features are sampled from two normal distributions with low variance. In particular, the uninformative features for \texttt{synth\_1} are sampled from a normal distribution with a standard deviation of 0.05. In contrast, the uninformative features for \texttt{synth\_2} are sampled from a normal distribution with a standard deviation of 0.5. 

\begin{table}[ht!]
    \scriptsize
    \begin{tabular}{c|c|c|c|c}
    \hline
    \multicolumn{3}{c|}{\texttt{synth\_1}} & \multicolumn{2}{c}{\texttt{synth\_2}}  \\ \hline
        s & HQFS  & ML-HQFS & HQFS & ML-HQFS \\ \hline
        2 & 0.05 & 0.37 & 0.21 & 0.20 \\
        3 & 0.05 & 0.35  & 0.21  & 0.20 \\ 
        4 & 0.05 & 0.37  & 0.21  & 0.20 \\ 
        5 & 0.15  & 0.99  & 0.27  & 0.20 \\  
        6 & 0.43 & 0.98  & 0.31  & 0.37 \\ \hline    
    \end{tabular}
    \caption{RBO measures on synthetic datasets.}
    \label{tab:result}
\end{table}
\begin{table}[ht!]
    \scriptsize
    \begin{tabular}{c|c|c|c|c}
    \hline
    \multicolumn{3}{c|}{\texttt{synth\_1}} & \multicolumn{2}{c}{\texttt{synth\_2}}  \\ \hline
        s & HQFS  & ML-HQFS & HQFS & ML-HQFS \\ \hline
        2 & 0.3 & 1.0 & 0.3 & 1.0 \\
        3 & 0.3 & 1.0 & 0.3 & 1.0 \\ 
        4 & 0.3 & 1.0 & 0.3 & 1.0 \\ 
        5 & 0.8 & 1.0 & 0.8 & 1.0 \\  
        6 & 0.7 & 1.0 & 0.7 & 1.0 \\ \hline    
    \end{tabular}
    \caption{Accuracy measure on synthetic datasets.}
    \label{tab:result_acc}
\end{table}

\subsubsection*{Results} In Table \ref{tab:result}, we report the HQFS and ML-HQFS results on \texttt{synth\_1} and \texttt{synth\_2} datasets. To mitigate the high computational cost of simulation, we test the algorithm with a maximum value of $s=6$. However, theoretical analysis suggests that increasing the value of $s$ can yield superior results. We note that ML-HQFS performs better than HQFS for \texttt{synth\_1} datasets. However, the RBO measure is unsuitable for the \texttt{synth\_2} dataset because the variance of the uninformative features is high, and then the algorithm fails to distinguish between informative and uninformative features.  In Table~\ref{tab:result_acc}, we report the accuracy achieved by our algorithm. We see that ML-HQFS always gets an accuracy of 100\% while the accuracy of HQFS ranges from 30\% to 80\% depending on the number $s$ of additional qubits.

\begin{figure}[ht!]
    \centering
    \includegraphics[width=90mm]{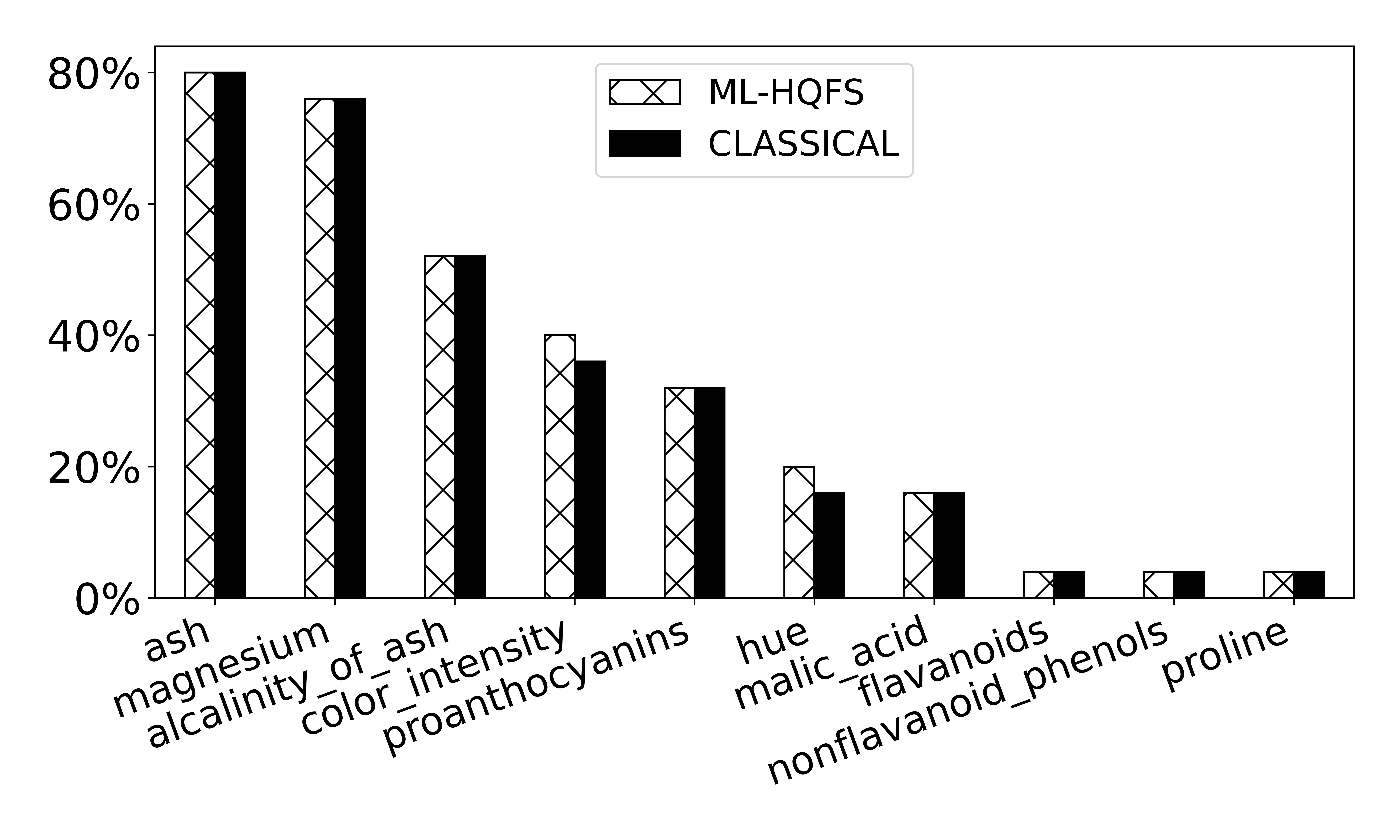}
    \caption{Features selected by ML-HQFS on \texttt{wine} dataset.}
    \label{fig:wine}
\end{figure}

\begin{figure}[ht!]
    \centering
    \includegraphics[width=90mm]{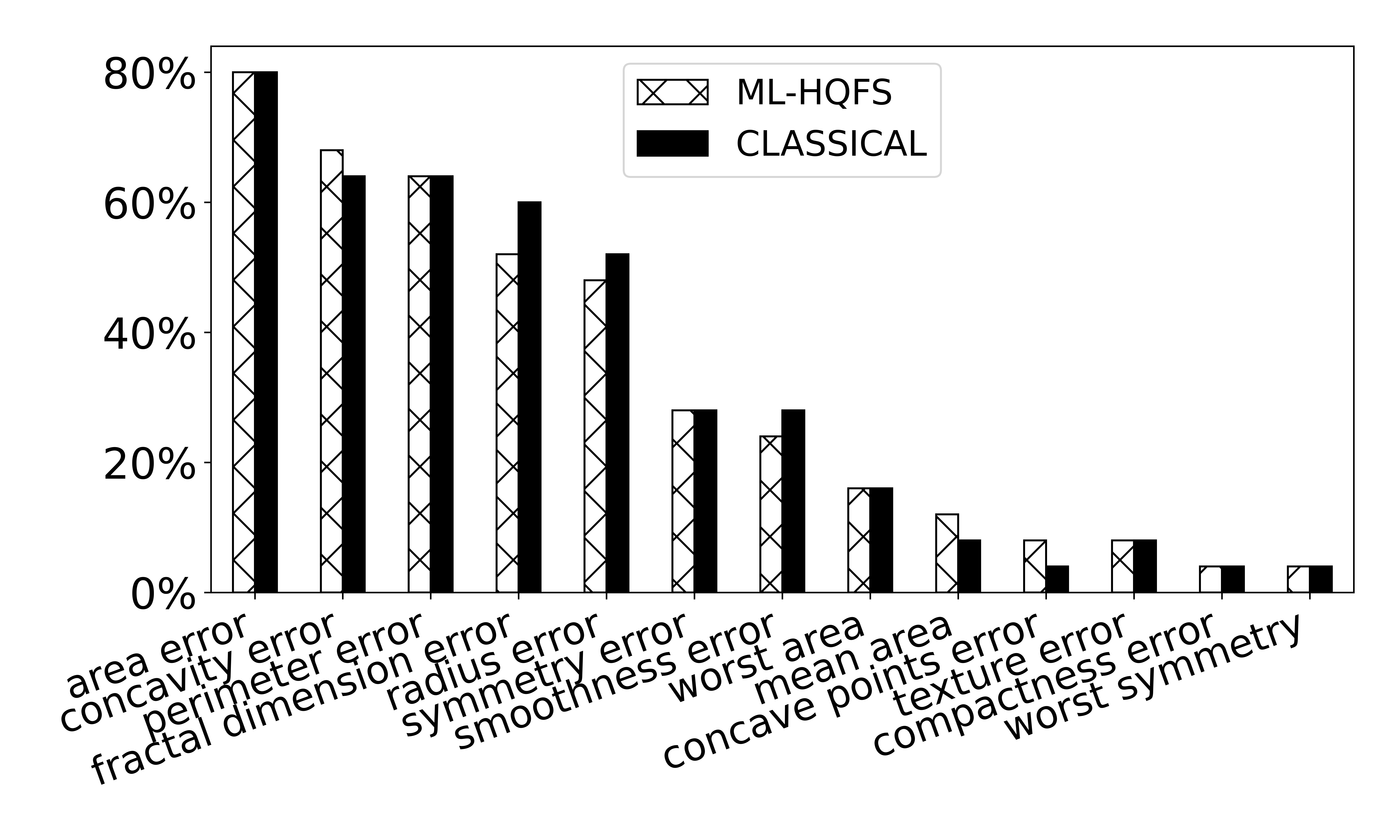}
    \caption{Features selected by ML-HQFS on \texttt{breast} dataset.}
    \label{fig:breast}
\end{figure}

For both real datasets, we randomly sample 16 records 25 times and execute the ML-HQFS algorithm with $s=6$ on these samples. We fix the threshold value as $t=0.1$ for the features of \texttt{wine} and $t=0.02$ for the features of \texttt{breast}. Then, we compare the percentage of times in which each feature is selected by the ML-HQFS to be discarded against the classical algorithm. Figure \ref{fig:wine} and Figure \ref{fig:breast}, show these comparisons for the \texttt{wine} and \texttt{breast} datasets, respectively. The figures show that both algorithms select the features with a quite similar frequency, with an MAE of 0.8\% for the \texttt{wine} dataset and an MAE of 2.15\% for \texttt{breast}.

\subsection{Experiments for Outlier Detection}
\label{sec:experiment_qoda}
We now present and discuss the numerical experiments conducted to evaluate the effectiveness of QODA on a series of benchmark datasets. The initial assessment involves testing the effectiveness of the heuristic used by QODA. The second experimental evaluation aims at assessing the accuracy of QODA using the \textsc{qasm simulator} of Qiskit. 

\subsubsection*{Dataset}
To evaluate the heuristic, we use four synthetic datasets generated with the PyOD library\footnote{\url{https://pyod.readthedocs.io/en/latest/index.html}} (\texttt{synth\_a}, \texttt{synth\_b}, \texttt{synth\_c}, and \texttt{synth\_d}), where normal data is generated by a multivariate Gaussian distribution and outliers are generated by uniform distribution, and three real datasets\footnote{\url{https://odds.cs.stonybrook.edu}} (\texttt{lympho}, \texttt{glass}, and \texttt{wbc}). They are defined as follows:

{\small
\begin{itemize}
\item \texttt{synth\_a}: M=500 synthetic data with dimension N=20 whose 2\% are outliers.
    \item \texttt{synth\_b}: M=500 synthetic data with dimension N=30 whose 2\% are outliers.
  
    \item \texttt{synth\_c}: M=500 synthetic data with dimension N=20 whose 10\% are outliers.
    \item \texttt{synth\_d}: M=500 synthetic data with dimension N=30 whose 10\% are outliers.
      \item \texttt{lympho}: real dataset with M=148, N=18.
    \item \texttt{glass}: real dataset with M=214, N=9.
    \item \texttt{wbc}: real dataset with M=278, N=30.
\end{itemize}}

\subsubsection*{Results}
We run two different kinds of experiments. The first kind is conducted to evaluate the effectiveness of the lower bound given by Proposition~\ref{prop2}. In Table~\ref{tab:errors}, we present the approximation errors introduced by the bound given in~\eqref{varineq}, measured in terms of Mean Squared Error (MSE), Mean Absolute Error (MAE), and Root Mean Squared Error (RMSE). Our findings indicate that the lower bound of Proposition~\ref{prop2} is very tight for all datasets considered.

\begin{table}[ht!]
\scriptsize
    \centering
    \begin{tabular}{l|ccc}
    \hline
        \textbf{dataset} & \textbf{MSE} & \textbf{MAE} & \textbf{RMSE} \\ \hline
        synth\_a & 3.02e-8 & 1.45e-4 & 1.73e-4 \\
        synth\_b & 1.92e-9 & 3.85e-5 &  4.38e-5 \\ 
        synth\_c & 2.53e-8 & 1.30e-4 & 1.60e-4 \\
        synth\_d & 1.57e-9 & 3.37e-5 & 3.96e-5 \\ \hline
        lympho & 3.83e-7 & 5.27e-4 & 6.19e-4 \\
        glass & 7.51-3 & 5.98e-2 & 8.66e-2 \\ 
        wbc & 3.09e-5 & 5.21e-3 & 5.55e-3 \\ \hline
    \end{tabular}
    \caption{Approximation errors of relation~\eqref{teo}}
    \label{tab:errors}
\end{table}

The second kind of experiment aims to evaluate the effectiveness of Outlier Detection leveraging the variance of the differences $\Delta^{(p)}_{ij}$ instead of the variance of the angles $\theta_{ij}^{(p)}$. We compare the outlier rankings generated by the two methods to achieve this objective. Table~\ref{tab:p-at-n} presents the P@$n$ metric, which we used to conduct this comparison. We can see that the precision is very high, especially for synthetic datasets. Furthermore, we utilized the RBO measure to evaluate the overall outlier ranking of our heuristic, taking as ground truth the classical algorithm working with the angles. The table indicates that the outlier rankings obtained by the method leveraging the variance of differences are similar to the ground truth, particularly for the highest-ranked outliers.

\begin{table}[ht!]
\scriptsize
    \centering
    \begin{tabular}{l|cccccc} \hline
       \textbf{dataset} & \textbf{n=5} & \textbf{n=10} & \textbf{n=15} & \textbf{n=20} & \textbf{n=25} & \textbf{n=30} \\ \hline
        synth\_a & 1.0 & 1.0 & 1.0 & 0.90 & 0.84 & 0.93 \\
        synth\_b & 1.0 & 1.0 & 0.93 & 0.95 & 0.84 & 0.90 \\ 
        synth\_c & 0.80 & 0.90 & 0.93 & 0.90 & 0.96 & 0.97 \\
        synth\_d & 1.0 & 1.0 & 1.0 & 0.95 & 0.96 & 0.93 \\ \hline
        lympho & 0.80 & 0.60 & 0.67 & 0.70 & 0.76 & 0.77 \\
        glass & 0.70 & 0.80 & 0.80 & 0.90 & 0.84 & 0.87 \\ 
        wbc & 0.60 & 0.70 & 0.80 & 0.70 & 0.60 & 0.57 \\ \hline
    \end{tabular}
    \caption{{\rm Precision-at-$n$ (P@$n$) of the algorithm taking as ground truth the classical algorithm working with the angles.}}
    \label{tab:p-at-n}
\end{table}

\begin{table}[ht!]
\scriptsize
    \centering
    \begin{tabular}{l|cccccc} \hline
        \textbf{dataset} & \textbf{p=0.70} & \textbf{p=0.75} & \textbf{p=0.80} & \textbf{p=0.85} & \textbf{p=0.90} & \textbf{p=0.95} \\ \hline
        synth\_a & 0.88 & 0.89 & 0.89 & 0.90 & 0.90 & 0.91 \\
        synth\_b & 0.65 & 0.69 & 0.74 & 0.79 & 0.84 & 0.88 \\ 
        synth\_c & 0.65 & 0.69 & 0.74 & 0.78 & 0.83 & 0.88 \\
        synth\_d & 0.67 & 0.71 & 0.76 & 0.81 & 0.86 & 0.91  \\ \hline
        lympho & 0.52 & 0.54 & 0.57 & 0.59 & 0.62 & 0.66  \\
        glass & 0.51 & 0.55 & 0.59 & 0.63 & 0.68 & 0.74  \\ 
        wbc & 0.60 & 0.63 & 0.66 & 0.68 & 0.68 & 0.67  \\ \hline
    \end{tabular}
    \caption{RBO measure varying the parameter $p$}
    \label{tab:rbo2}
\end{table}

Currently, the availability of quantum hardware is limited, which means that quantum computation can only be simulated using classical hardware. Nevertheless, simulating quantum computation on classical hardware is expensive, so we can only execute basic versions of our algorithm. For this reason, we run QODA on two simple synthetic datasets with a single outlier each. The first has $M=5$ records with dimension $N=1$, and the second has $M=5$ records with dimension $N=3$. 
By running the experiments with different random seed values for generating datasets five times, we found that QODA successfully identified the sole outlier in both datasets with a 100\% accuracy rate.

\section{Conclusion}\label{sec:conclusion}
In this study, we proposed a novel quantum algorithm called QVAR to estimate the variance within a set of values while in a superposition state. We utilized the QVAR subroutine in two prominent and crucial areas within data analysis: Feature Selection and Outlier Detection. This was accomplished by creating two hybrid approaches to tackle these tasks: the Hybrid Quantum Feature Selection (HQFS) algorithm and the Quantum Outlier Detection Algorithm (QODA). 

Our experiments revealed that when the extra qubits are chosen properly, QVAR effectively computes an estimation of the classical variance. Consequently, concerning the Feature Selection task, the feature rankings generated by HQFS closely resemble those derived from classical algorithms, particularly for low-variance features. This suggests that HQFS correctly discards uninformative features, showcasing its accuracy in feature selection. The experiments for QODA indicate that the heuristic we proposed yields results comparable to those of the classical ABOD algorithm. Overall, our hybrid algorithm offers a promising approach for outlier detection tasks.

The algorithms we introduced are intended to process classical data as input. Nonetheless, the QVAR subroutine holds potential for broader application in scenarios involving an initial state of quantum values in superposition. We consider the prospect of leveraging the QVAR subroutine for other tasks as a direction for future research.

\backmatter

\subsection*{Acknowledgement}
This study was carried out within the National Centre on HPC, Big Data and Quantum Computing - SPOKE 10 (Quantum Computing) and received funding from the European Union Next-GenerationEU - National Recovery and Resilience Plan (NRRP) - MISSION 4 COMPONENT 2, INVESTMENT N. 1.4 - CUP N. I53C22000690001. This manuscript reflects only the authors' views and opinions, neither the European Union nor the European Commission can be considered responsible for them. Partial support was also provided by INdAM-GNCS.
This work is also supported by the European Community under the Horizon~2020 programme G.A. 871042 \emph{SoBigData++} and by the NextGenerationEU programme under the funding scheme ``\textit{SoBigData.it} - Strengthening the Italian RI for Social Mining and Big Data Analytics'' - Prot. IR0000013.

\bibliography{sn-bibliography}

\end{document}